\documentclass[journal,twoside,web]{ieeecolor}  

\usepackage{lcsys}

\usepackage{mathtools}
\mathtoolsset{showonlyrefs,showmanualtags}

\usepackage{amsfonts}
\usepackage{amsmath}
\usepackage{amssymb}
\usepackage{graphicx}

\usepackage{algorithm2e}

\usepackage{siunitx}

\makeatletter
\let\MYcaption\@makecaption
\makeatother

\usepackage[font={sf,footnotesize}]{subcaption}

\makeatletter
\let\@makecaption\MYcaption
\makeatother

\usepackage{etoolbox}
\makeatletter
\@ifundefined{color@begingroup}%
  {\let\color@begingroup\relax
   \let\color@endgroup\relax}{}%
\def\fix@ieeecolor@hbox#1{%
  \hbox{\color@begingroup#1\color@endgroup}}
\patchcmd\@makecaption{\hbox}{\fix@ieeecolor@hbox}{}{\FAILED}
\patchcmd\@makecaption{\hbox}{\fix@ieeecolor@hbox}{}{\FAILED}

\usepackage{tikz}
\usetikzlibrary{shapes,arrows}

\newtheorem{definition}{Definition}[section]
\newtheorem{theorem}{Theorem}[section]
\newtheorem{lemma}{Lemma}[section]
\newtheorem{remark}{Remark}[section]
\newtheorem{assumption}{Assumption}[section]

\usepackage{hyperref}

\usepackage{xparse}

\NewDocumentCommand\bbm{}{ \begin{bmatrix} }
\NewDocumentCommand\ebm{}{ \end{bmatrix} }
\NewDocumentCommand\Vector{m}{ \boldsymbol{\mathbf{#1}} }

\NewDocumentCommand\Diag{m}{ \mathrm{diag} \left\{ #1 \right\} }

\NewDocumentCommand\Real{}{ \mathbb{R} }
\NewDocumentCommand\Integers{}{ \mathbb{N} }
\NewDocumentCommand\Comp{}{ \mathcal{K} }

\NewDocumentCommand\Data{}{ \mathcal{D} }

\NewDocumentCommand\Expectation{m}{ \mathbb{E}\left[#1\right] }
\NewDocumentCommand\NormalDist{mm}{ \mathcal{N}\left(#1,#2\right) }

\NewDocumentCommand\Prior{m}{\tilde{#1}}
\NewDocumentCommand\Mean{m}{\pmb{\mu}^{m}}
\NewDocumentCommand\mean{}{\mu}
\NewDocumentCommand\Cov{m}{\pmb{\Sigma}^{m}}
\NewDocumentCommand\cov{}{\sigma}

\NewDocumentCommand\X{}{ \Vector{x} }
\NewDocumentCommand\Xd{}{ \dot{\Vector{x}} }

\NewDocumentCommand\xd{}{ \dot{x} }
\NewDocumentCommand\xdd{}{ \ddot{x} }

\NewDocumentCommand\Z{}{ \Vector{z} }
\NewDocumentCommand\Zref{}{ \Vector{z}^{\text{ref}}}
\NewDocumentCommand\Zd{}{ \dot{\Vector{z}} }
\NewDocumentCommand\Zhat{}{ \hat{\Vector{z}} }

\NewDocumentCommand\Inv{m}{ {#1}^{-1} }
\DeclareMathOperator{\Prob}{\text{Pr}}

\DeclareMathOperator{\Ad}{\mathbf{A}_d}
\DeclareMathOperator{\Bd}{\mathbf{B}_d}
\DeclareMathOperator{\A}{\mathbf{A}}
\DeclareMathOperator{\B}{\mathbf{B}}
\DeclareMathOperator{\K}{\mathbf{K}}
\DeclareMathOperator{\Q}{\mathbf{Q}}
\NewDocumentCommand\R{}{r}
\DeclareMathOperator{\Pv}{\mathbf{P}}

\NewDocumentCommand\Hcon{}{\mathbf{h}}

\DeclareMathOperator{\Zcon}{\mathcal{Z}}
\DeclareMathOperator{\Ucon}{\mathcal{U}}

\newcommand{\umin}{u_{\text{min}}}
\newcommand{\umax}{u_{\text{max}}}

\NewDocumentCommand\e{}{\mathbf{e}}
\NewDocumentCommand\PSI{}{\psi(\Z^*_k,u_k)}
\NewDocumentCommand\vnom{}{v^{\text{nom}}}
\NewDocumentCommand\vref{}{v^{\text{ref}}}
\NewDocumentCommand\vdes{}{v^*}

\newcommand{\Horizon}{N}
\newcommand{\MPCLoss}{J_{\Horizon}}
\newcommand{\timestep}{\delta_t}

\NewDocumentCommand{\Feature}{}{\Vector{a}}
\NewDocumentCommand{\gam}{mmo}{\gamma^{#3}_{#1}(#2)}
\NewDocumentCommand{\kap}{mm}{\mathbf{k}_{#1} (#2)}
\NewDocumentCommand{\kapT}{mm}{\mathbf{k}^T_{#1} (#2)}
\NewDocumentCommand{\Kap}{}{\mathbf{K}}
\NewDocumentCommand\SafeProb{}{\delta}

\NewDocumentCommand{\PRS}{}{\mathcal{R}}
\NewDocumentCommand{\PRSerror}{}{\Delta}
\NewDocumentCommand{\quantile}{}{\rho}

\NewDocumentCommand{\Asoc}{}{\bar{\mathbf{A}}}
\NewDocumentCommand{\bsoc}{}{\bar{\mathbf{b}}}
\NewDocumentCommand{\csoc}{}{\bar{\mathbf{c}}}
\NewDocumentCommand{\dsoc}{}{\bar{\mathbf{d}}}
\NewDocumentCommand{\usoc}{}{\bar{\mathbf{u}}}
\newcommand{\new}[1]{#1}

\title{\LARGE \bf
Differentially Flat Learning-Based Model Predictive Control Using a Stability, State, and Input Constraining Safety Filter
}

\author{Adam W. Hall$^{1\;\dagger}$, Melissa Greeff$^{2\;\dagger}$, and Angela P. Schoellig$^{3\;\dagger}$
\thanks{$^{1}$Adam W. Hall is jointly with the Learning Systems and Robotics Lab (\href{www.learnsyslab.org}{\text{www.learnsyslab.org}}) and the STARS lab (\href{https://starslab.ca/ }{starslab.ca}) at the University of Toronto Institute for Aerospace Studies (UTIAS), Toronto, Canada. Email: adam.hall@robotics.utias.utoronto.ca}%
\thanks{$^{2}$Melissa Greeff is with the Robora Lab (\href{www.roboralab.com}{\text{www.roboralab.com}}) at the Department of Electrical and Computer Engineering, Queen's University, Kingston, Canada.}
\thanks{$^{3}$Angela P.~Schoellig is with the Learning Systems and Robotics Lab (\href{www.learnsyslab.org}{\text{www.learnsyslab.org}}) at the Technical University of Munich and the University of Toronto, and the Munich Institute for Robotics and Machine Intelligence (MIRMI), Munich, Germany.}
\thanks{$^{\dagger}$All authors are with the Vector Institute for Artificial Intelligence, Toronto, Canada.}
}

\begin{document}
\bstctlcite{IEEEexample:BSTcontrol}

\maketitle
\thispagestyle{empty}
\thispagestyle{empty}
\pagestyle{empty}

\begin{abstract}
Learning-based optimal control algorithms control unknown systems using past trajectory data and a learned model of the system dynamics.
These controllers use either a linear approximation of the learned dynamics, trading performance for faster computation, or nonlinear optimization methods, which typically perform better but can limit real-time applicability.
In this work, we present a novel nonlinear controller that exploits differential flatness to achieve similar performance to state-of-the-art learning-based controllers but with significantly less computational effort.
Differential flatness is a property of dynamical systems whereby nonlinear systems can be exactly linearized through a nonlinear input mapping.
Here, the nonlinear transformation is learned as a Gaussian process and is used in a safety filter that guarantees, with high probability, stability as well as input and flat state constraint satisfaction.
This safety filter is then used to refine inputs from a flat model predictive controller to perform constrained nonlinear learning-based optimal control through two successive convex optimizations.
We compare our method to state-of-the-art learning-based control strategies and achieve similar performance, but with significantly better computational efficiency, while also respecting flat state and input constraints, and guaranteeing stability.
\end{abstract}


\vspace{-1.0em}
\section{INTRODUCTION}
\IEEEPARstart{I}{n} recent years, interest has grown in controlling safety-critical systems whose dynamics are partially unknown, like unmanned aerial vehicles, driverless cars, and mobile manipulators.
\new{Classically, guaranteeing safety and stability of these uncertain systems results in overly conservative behaviour, limiting their usage in real tasks.}
Using machine learning and past trajectory data inside classical control frameworks \new{to learn the system dynamics} has proven to be an effective, safe learning-based control technique, \new{but often requires slow nonlinear optimizations and can suffer from poor computational stability and efficiency} \cite{brunkeSafeLearningRobotics2022}.
For example, Gaussian process model predictive control (GPMPC) uses a Gaussian process (GP) to model the uncertain dynamics.
This learned model is then used inside a robust model predictive control (MPC) framework.
GPMPC has been used on mobile robots \cite{Ostafew2016a}, quadrotors \cite{torrenteDataDrivenMPCQuadrotors2021}, and other autonomous systems. 

A drawback of GPMPC is its computational complexity.
It either requires powerful on-board computation or remote computation of inputs, limiting its use in real systems.
Standard GPs require all of their training data to be stored in memory and their posterior mean and covariance predictions are computationally expensive, even when using approximate methods \cite{quinonero-candelaUnifyingViewSparse2005}---many approximations and `tricks' are used to achieve real-time operability.
There are, however, structural assumptions about the dynamics that can improve computational speed without sacrificing performance, such as incorporating a control-affine structure, and differential flatness.

Differential flatness is a property of many nonlinear dynamical systems that enables their transformation into linear systems through a nonlinear input mapping, called exact linearization \cite{fliessFlatnessDefectNonlinear1995}---this is an exact transformation, not an approximation.
Linear control techniques can then be used to compute a flat input that is subsequently transformed through this nonlinear mapping and applied to the real system.
Many real robotic systems are differentially flat, like quadrotors \cite{Greeff2018}, flexible-joint manipulators \cite{isidoriNonlinearControlSystems1995}, and mobile robots \cite{Sira-Ramirez2004}, to list a few.
In the literature, differential flatness has been exploited to control nonlinear systems using linear MPC \cite{Greeff2018}, called flat MPC (FMPC), \new{which achieves similar performance to nonlinear MPC while greatly improving computational efficiency.}

\begin{figure}
\centering
\includegraphics[width=\linewidth]{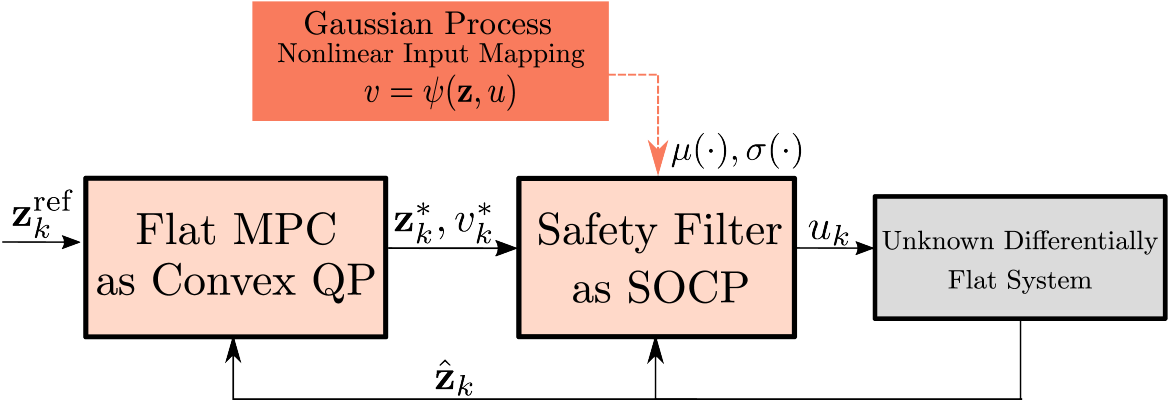}
\caption{
Our proposed architecture enables high-performance trajectory tracking for uncertain differentially flat systems by developing a learning-based model predictive controller that is computationally efficient to compute. 
We do this by solving two convex optimization problems: 1) a convex quadratic program in flat model predictive control that finds the flat input $v^*_k$ and state $\Z^*_k$ to track a reference $\Zref_k$; 2) a second-order cone program in the safety filter that uses a learned representation of the flat nonlinear input mapping $v = \psi(\Z,u)$, with posterior prediction mean $\mu(\cdot)$ and covariance $\sigma(\cdot)$, to perform probabilistic feedback linearization guaranteeing asymptotic stability, as well as state and input constraint satisfaction.  
  \label{fig:overview}} 
  \vspace{-2.0em}
\end{figure}



Differential flatness has also been exploited in safe learning-based controllers.
In \cite{Greeff2021}, the nonlinear input mapping and its inverse are learned from trajectory data and used in a robust linear quadratic regulator \new{(RLQR)} formulation to control an uncertain system, while guaranteeing an ultimate upper bound on tracking error. 
\new{This method is more computationally efficient than GPMPC}, however, still involves a demanding nonlinear optimization and cannot explicitly handle input constraints \new{or state constraints}.

More recently, \cite{greeffLearningStabilityFilter2021} has used a learned representation of the nonlinear input mapping inside of a safety filter that finds the input that best matches the desired input from a flat linear controller while guaranteeing probabilistic asymptotic stability.
In particular, the safety filter is formulated as a Second-Order Cone Program (SOCP) by exploiting differential flatness and the affine-control structure of the dynamics.
This SOCP \new{is more computationally efficient than GPMPC and RLQR} and can explicitly handle input constraints, but it cannot enforce any state constraint guarantees.

In this work, we build upon \cite{greeffLearningStabilityFilter2021} to provide three main contributions:
\begin{itemize}
    \item a novel safe learning-based MPC for nonlinear, differentially flat, control-affine systems that can be solved via two successive convex optimizations;
    \item a novel safety filter that guarantees, with high probability, \new{asymptotically stable} tracking error, and flat-state and input constraint satisfaction, modelled as an SOCP;
    \item a comparison with GPMPC, in simulation, that shows our approach achieves similar performance but is at least 10 times more computationally efficient.
\end{itemize}
\new{The advantage of MPC, used in this work, over linear quadratic regulation (LQR) used in \cite{greeffLearningStabilityFilter2021}, is anticipating the reference and constraint boundaries which avoids infeasible states and aggressive inputs, as shown in our simulated examples.} 
All code is available at \href{https://github.com/utiasDSL/fmpc_socp}{\text{github.com/utiasDSL/fmpc\_socp}}.

\section{Problem Statement} \label{sec:ps}
We consider a single-input continuous-time control-affine nonlinear system
\begin{equation} \label{eq:nonlin_dyn}
    \Xd(t) = f(\X(t)) + g(\X(t))u(t),
\end{equation}
with initial condition $\X(0) = \X_0$, where $\X(t) \in \Real^n$ and $u(t) \in \Real$ are the state and input of the system at time $t \in \Real_{\geq 0}$.
The maps $f : \Real^n \rightarrow \Real^n$ and $g : \Real^n \rightarrow \Real^n$ are unknown, but are assumed to be locally Lipschitz continuous.
It is assumed that a prior model of the unknown system is given by
\begin{equation} \label{eq:prior}
\dot{\mathbf{x}}(t) = \Prior{f}(\X(t)) + \Prior{g}(\X(t)) u(t).
\end{equation}

\begin{assumption} \label{as:diff_flat}
The system \eqref{eq:nonlin_dyn} and the prior model \eqref{eq:prior} \new{are single-input single-output systems} and differentially flat with respect to \new{a}  \textit{known} flat output $y = h(\X(t))$, with $y \in \Real$.
\end{assumption}

\begin{definition}[Differential Flatness \cite{fliessFlatnessDefectNonlinear1995}]
    A smooth single input nonlinear system is \textit{differentially flat} if there exists a \textit{flat output} $y \in \Real$ such that the input and all states can be uniquely determined from this \textit{flat output} and its derivatives $\Z = [y, \dot{y}, \ldots, y^{(n-1)}]^T$, and there exists smooth, invertible functions $\X = \Inv{\phi}(\Z)$, $u = \Inv{\psi}(\Z, v)$, where $v = y^{(n)}$.
\end{definition}

\begin{lemma}[Linearized Flat Dynamics \cite{ fliessFlatnessDefectNonlinear1995, isidoriNonlinearControlSystems1995}]
    A differentially flat system \eqref{eq:nonlin_dyn} can be transformed into a linear system in the Brunowsky Canonical form
    \begin{equation} \label{eq:ct_brunowsky}
        \Zd(t) = \A\Z(t) + \B v(t),
    \end{equation}
    where the new input $v$ is related to the system input via the nonlinear mapping
    \begin{equation}
        v = \psi(\Z,u). \label{eq:v_from_u}
    \end{equation}
    Furthermore, if the system is control affine, then this mapping has the specific form
    \begin{equation}
        v = \alpha(\Z) + \beta(\Z)u. \label{eq:v_from_u_a_b}
    \end{equation}
\end{lemma}
\begin{remark}
Given that the mappings $f$ and $g$ in \eqref{eq:nonlin_dyn} are unknown, $\alpha$ and $\beta$ are also unknown.
\end{remark}

Our system \eqref{eq:nonlin_dyn} is subject to input constraints $\Ucon \coloneqq \{ u \in \Real \; | \; \umin \leq u \leq \umax \}$ and convex constraints on the flat-state vector $\Z \in \Zcon$. 

\begin{remark}
The constraints are enforced on the flat state $\Z$.
While the flat state differs from the state $\X$, in many robotic systems it still represents physical quantities that are constrained.
For example, in quadrotors, it represents the position, velocity, and acceleration; in flexible joint manipulators, it represents the output shaft joint angle and its derivatives \cite{isidoriNonlinearControlSystems1995}; and for mobile robotics, the commonly used unicycle model's flat state comprises position, velocity, and acceleration \cite{Sira-Ramirez2004}.
\end{remark}

The objective is to design a computationally efficient controller for the unknown system \eqref{eq:nonlin_dyn} that achieves high tracking performance of a reference trajectory $\Zref$, guarantees tracking convergence with high probability, and respects input and state constraints.

\vspace{-1.0em}
\section{Background} \label{sec:bk}

\subsection{Discrete-Time Control Lyapunov Function} \label{def:clf}
Given a \new{constant} sampling period $\timestep$, the discretization of \eqref{eq:ct_brunowsky} becomes
\begin{equation}
\Z_{k+1} = \Ad \Z_k + \Bd v_k, \label{eq:dt_drunowsky}
\end{equation}
where $\Ad$ and $\Bd$ are \new{Euler} discretizations of $\A$ and $\B$.
Further, $\Z_k = z(\timestep k)$ and $v_k = v(\timestep k)$ are time-sampled flat states and inputs at time step \new{$k \in \Integers_{\geq 0}$}.

Consider the smooth reference $\Zref(t) : [0,T) \rightarrow \Real^n$ and $ \vref(t) : [0,T) \rightarrow \Real$ sampled every $k\timestep$ to yield the discrete reference signals $\Zref_k = \Zref(k\timestep)$, $\vref_k = \vref(k\timestep)$.
The tracking error can be defined with respect to this reference as $\e_k = \Z_k - \Zref_k$.
\new{Given an error feedback control policy $v_k = -\K \e_k + \vref_k$ with gain $\K \in \Real^{1 \times n}$, the error dynamics become $\e_{k+1} = (\Ad - \Bd\K)\e_k$.}

\begin{definition}
If a function $V : \Real^n \rightarrow \Real_{\geq 0}$ satisfies
\begin{gather}
V(0) = 0 \text{ and } V(\e_k) > 0,\; \forall \e_k \in \Zcon \setminus \{ 0 \} \label{eq:lyap_cond1} \\
V(\e_{k+1}) \new{<} V(\e_k),\; \forall \e_k \in \Zcon \setminus \{0 \} \label{eq:lyap_decrease}
\end{gather}
for the error dynamics, then it is called a \textit{Control Lyapunov Function} (CLF) and its existence guarantees the asymptotic stability of the closed-loop dynamics.
\end{definition}

\begin{lemma}[\cite{isidoriNonlinearControlSystems1995}]
Given Assumption \ref{as:diff_flat}, if the transformation \eqref{eq:v_from_u} is known, $u_k$ can be chosen to cancel the nonlinear term \eqref{eq:v_from_u} and find a $v_k$ such that the resulting linear error dynamics are Hurwitz.
\end{lemma}
\begin{remark}
The true dynamics are unknown, and thus the inverse transform $u = \psi(\Z, v)$ is unknown and is often approximated using the prior model \eqref{eq:prior}.
Such control methods, however, are not robust to model mismatch in \eqref{eq:v_from_u} \cite{Greeff2018}.
\end{remark}



\subsection{Gaussian Processes (GPs)} \label{sec:gp}
GPs are used to model nonlinear functions $\psi(\Feature) : \Real^{\dim(\Feature)} \rightarrow \Real$. 
They encapsulate a prior over possible functions. 
As data are collected, the possible functions for $\psi(\Feature)$ are refined and the GP obtains a posterior distribution over functions \cite{rasmussenGaussianProcessesMachine2006}. 
GPs assume that all collected data is jointly Gaussian with a prior mean and covariance.
GP regression is the process of finding the hyperparameters that optimize the log-likelihood of the marginal distribution over the sampled function data. 

Given a query point $\Feature^*$, the posterior prediction conditioned on the data $\Data = \{(\Feature_i,\hat{\psi}(\Feature_i) \}_{i=1}^{N_{\Data}}$ is given by the distribution $\psi(\Feature^*) | \Data \sim \NormalDist{\mean(\Feature^*)}{\cov^2(\Feature^*)}$.
Here, the posterior mean and covariance are 
$\mean(\Feature^*) = \kap{}{\Feature^*} \Inv{\Kap} \hat{\Psi}$, and
$\cov(\Feature^*) = k(\Feature^*, \Feature^*) - \kap{}{\Feature^*} \Inv{\Kap} \kapT{}{\Feature^*}$,
where $\kap{}{\Feature^*} = [k(\Feature^*,\Feature_i), \ldots , k(\Feature^*,\Feature_{N_\Data})]$, $ \Kap_{i,j} = k(\Feature_i, \Feature_j)$, and $\hat{\Psi} = [\hat{\psi}(\Feature_1), \ldots, \hat{\psi}(\Feature_{N_\Data})]^T$. See \cite{rasmussenGaussianProcessesMachine2006} for more details.

\vspace{-1.0em}
\section{Methodology}
In our proposed architecture, shown in Figure \ref{fig:overview}, we exploit the control-affine and differentially flat structure of the system \eqref{eq:nonlin_dyn} to design a learning-based model predictive controller. 
Our method solves two convex optimization problems at each time-step: a convex FMPC to determine the flat input to the linear dynamics \eqref{eq:dt_drunowsky} that best tracks a given reference trajectory; and a safety filter formulated as an SOCP that ensures probabilistic asymptotic stability, probabilistic flat-state constraint satisfaction, and input constraint satisfaction, even when \eqref{eq:nonlin_dyn} is uncertain.
By leveraging system data, we develop a controller that achieves high tracking performance despite unknown system dynamics \eqref{eq:nonlin_dyn}, but is still computationally efficient, \new{making onboard input computation more practical relative to the current state-of-the-art}.

This section follows the components of our controller given in Figure \ref{fig:overview}.
First, in Section \ref{sec:FMPC}, the FMPC is formulated, assuming the system dynamics \eqref{eq:nonlin_dyn} and the input mapping $\new{u = \psi^{-1}(\Z,v)}$ are known.
As \eqref{eq:nonlin_dyn} is not known, a learned representation of \eqref{eq:v_from_u} using GPs is used.
In particular, our approach exploits the affine structure of \eqref{eq:nonlin_dyn} and \eqref{eq:v_from_u}, as detailed in Section \ref{sec:aff_gp}.
Finally, in Section \ref{sec:sf}, the probabilistic feedback linearization (Section \ref{sec:prob_fbl}), probabilistic asymptotic stability (Section \ref{sec:stab}), and probabilistic flat-state constraint (Section \ref{sec:state}), are formulated as an SOCP safety filter (Section \ref{sec:socp}).


\vspace{-1.0em}
\subsection{Flat Model Predictive Control} \label{sec:FMPC}

FMPC \new{iteratively} solves a convex finite-horizon optimal control problem (OCP) to control the nonlinear flat system \eqref{eq:nonlin_dyn}.
In this section, the FMPC formulation is presented and builds on the FMPC formulation in \cite{Greeff2018}.
We highlight how FMPC can be designed such that the error dynamics, with respect to a flat reference $\Zref$, are asymptotically stable when the system dynamics \eqref{eq:nonlin_dyn} are known.


\begin{assumption}[\cite{gruneNonlinearModelPredictive2017}]\label{as:cost_bound}
There exists a cost function $\ell$ that is bounded by comparison functions $\bar{\zeta},\; \underline{\zeta} \in 
 \Comp_{\infty}$ such that
$\underline{\zeta}(\| \Z_k - \Zref_k\| ) \leq \inf_{v_k} \ell(\Z_k, \Zref_k,v_k) \leq \bar{\zeta}(\| \Z_k - \Zref_k \| ) \; \forall \; \Z_k \in \Zcon$.
\end{assumption}

\begin{remark}
We consider cost functions of the form $\ell(\Z_k, \Zref_k ,v_k) = (\Z_k - \Zref_k)^T \Q (\Z_k - \Zref_k) + v_k^2 \R $, where $\Q \in \Real^{n \times n},\; \Q \succ 0,\; \R \in \Real,\; \R > 0$.
\end{remark}

The convex finite-horizon OCP solved at each time step is given by
\begin{align} 
&\min_{\Z_{k|k:k+{\Horizon}}, v_{k|k:k+{\Horizon}-1}}   \MPCLoss(\Z_{k|k+1:k+{\Horizon}}, \Zref_{k|k+1:k+\Horizon}, v_{k|k:k+\Horizon-1}) \nonumber \\
&\text{s.t.} \;\;\; \label{eq:mpc}  \Z_{k|k} = \Zhat_k \\
& \qquad \Z_{k|i+1} = \Ad \Z_{k|i} + \Bd v_{k|i}, \; \forall i \in [k, k+\Horizon-1]  \\
& \qquad \Z_{k|i} \in \Zcon\;  \forall i \in [k, k+\Horizon], \nonumber 
\end{align}
where $\Horizon \in \Integers$ is the horizon length, $\hat{\Z}_k$ is the measured flat state at time step $k$,
$\MPCLoss(\cdot,\cdot , \cdot) = \sum_{i=k+1}^{k+{\Horizon}} \ell(\Z_{k|i}, \Zref_{k|i},v_{k|i-1})$
is the cost function, and the flat states $\Z_{k|k:k+{\Horizon}} = [ \Z_k, \ldots, \Z_{k+\Horizon}]$, the flat reference $\Zref_{k|k:k+{\Horizon}} = [ \Zref_k, \ldots, \Zref_{k+\Horizon}]$, and the flat inputs $v_{k|k:k+{\Horizon}-1} = [v_k, \ldots, v_{k+{\Horizon}-1}]$ are sequences for time step $k$ to time step $k+N$, computed at time step $k$.
The optimal solution to OCP \eqref{eq:mpc} at time step $k$ is given by the sequences $\Z^*_{k|k:k+{\Horizon}}$ and $v^*_{k|k:k+{\Horizon}-1}$.
We use the notation $\Z^*_k \coloneqq \Z^*_{k|k}$ and $v^*_k \coloneqq v^*_{k|k}$ for the optimal flat state and flat input computed at time step $k$. 
This state-input pair is used in the safety filter in Section \ref{sec:sf}.
This OCP then applies the zero-order hold input $u(t) = \Inv{\psi}(\Z^*_k,v^*_k) \; \forall \; t \in [k\timestep,(k+1)\timestep)$ to the system.


\begin{assumption}[\cite{gruneNonlinearModelPredictive2017}] \label{as:vf_bound}
There exists a comparison function $\xi \in \Comp_{\infty}$ and an integer $s \in \Integers$ such that, for all $\Z \in \Zcon$, the inequality 
$\inf_{v} J_{s}(\Z, \Zref, v) \leq \xi(\inf_{v} \ell(\Z, \Zref, v))$
holds for all $s \in \Integers$.
\end{assumption}


\begin{lemma}[{\new{\cite[Thm. 6.2]{gruneNonlinearModelPredictive2017}, \cite[Sec. 8.3]{borrelli2017predictive}}}] \label{lem:mpc}
Given Assumptions \ref{as:cost_bound} and \ref{as:vf_bound}, there exists \new{an} $\Horizon \in \Integers$ such that the sampled-data MPC, defined by OCP \eqref{eq:mpc}, is \textit{asymptotically stable} in the closed-loop with respect to the reference $\Zref$.
Furthermore, in the unconstrained case (i.e., where the flat-state constraint $\Z_i \in \Zcon $ is not present), the OCP \eqref{eq:mpc} permits an equivalent closed-form solution 
\begin{equation} \label{eq:fmpc_fb}
v^*_k = \new{-\K} (\Z_k - \Zref_k) + \vref_k,
\end{equation}
where $\K \in \Real^{1 \times n}$ is the equivalent gain matrix.
\end{lemma}
\vspace{-1.0em}
\subsection{Gaussian Process Learning} \label{sec:aff_gp}
FMPC, presented in the previous section, relies on knowing \eqref{eq:nonlin_dyn} to compute the input $u_k = \Inv{\psi}(\Z^*_k,v^*_k)$ from the optimized trajectory (i.e., $\Z^*_k$ and $v^*_k$).
Given \eqref{eq:nonlin_dyn} is unknown, we propose to learn the map \eqref{eq:v_from_u} as a GP.
We thus encode the control-affine structure of \eqref{eq:v_from_u_a_b} in the kernel selection of the GP.
This structure enables us to formulate probabilistic stability and state constraints in our filter in Section \ref{sec:sf} as an SOCP.
Specifically, we select the following kernel \cite{castanedaGaussianProcessbasedMinnorm2021}
\begin{equation} \label{eq:aff_kern}
    k(\Feature_i,\Feature_j) = k_\alpha(\Z_i,\Z_j) +  u_i k_\beta(\Z_i,\Z_j) u_j + \delta_{i,j}\sigma_\eta^2.
\end{equation}
\begin{assumption} \label{ass:pd_kern}
$k_\alpha(\cdot,\cdot)$ and $k_\beta(\cdot, \cdot)$ are positive definite and bounded kernels.
\end{assumption}
\begin{lemma} \cite[Lem. 2]{greeffLearningStabilityFilter2021}
    Given Assumption \eqref{ass:pd_kern}, the affine kernel in \eqref{eq:aff_kern} is also bounded and positive definite.
\end{lemma}
 Given a query point $\Feature = (\Z, u)$ and a regressed GP conditioned on $N_{\Data}$ noisy observations $\Data = \{ \Feature_i, \hat{\psi}(\Feature_i) \}_{i=1}^{N_{\Data}}$, the posterior mean prediction $\mean(\Feature)$ and variance $\cov(\Feature)$ are linear and quadratic in $u$, respectively, 
\begin{gather}
\mean(\Feature) = \gam{1}{\Z}[] + \gam{2}{\Z}[]u, \label{eq:aff_mean}\\
\cov^2(\Feature) = \gam{3}{\Z}[] + \gam{4}{\Z}[]u + \gam{5}{\Z}[]u^2, \label{eq:aff_cov}
\end{gather}
 where
 \begin{align}
\gam{1}{\Z}[] &= \kap{\alpha}{\Z} \Inv{\Kap} \hat{\mathbf{\Psi}}, \qquad \gam{2}{\Z}[] = \kap{\beta}{\Z} \Inv{\Kap} \hat{\mathbf{\Psi}}, \\
\gam{3}{\Z}[] &= k_\alpha(\Z, \Z) - \kap{\alpha}{\Z} \Inv{\Kap} \kapT{\alpha}{\Z}, \\
\gam{4}{\Z}[] &= -(\kap{\beta}{\Z} \Inv{\Kap} \kapT{\alpha}{\Z}  + \kap{\alpha}{\Z} \Inv{\Kap} \kapT{\beta}{\Z}),\\
\gam{5}{\Z}[] &= k_\beta(\Z,\Z) - \kap{\beta}{\Z} \Inv{\Kap} \kapT{\beta}{\Z}.
 \end{align}
Here, $\hat{\Psi} \in \Real^{N_\Data}$, with $\hat{\Psi}_i = \hat{\psi}(\Feature_i)$,
$\kap{\alpha}{\Z} \in \Real^{1 \times N_\Data}$ with $\kap{\alpha,i}{\Z} = k_\alpha(\Z,\Z_i)$, 
$\new{\kap{\beta}{\Z}} \in \Real^{1 \times N_\Data}$ with $\kap{\beta,i}{\Z} = k_\beta(\Z,\Z_i)$,
 and $\Kap \in \Real^{N_{\Data} \times N_{\Data}}$ with elements $ \Kap_{i,j} = k(\Feature_i, \Feature_j)$.

\vspace{-1.0em}
\subsection{Safety Filter} \label{sec:sf}

In this section, we use the learned GP model of \eqref{eq:v_from_u_a_b} in a safety filter design with three components: 1) probabilistic feedback linearization, 2) a probabilistic stability constraint, and 3) a probabilistic state constraint.
We then show how to implement the safety filter as a SOCP.  

\subsubsection{Probabilistic Feedback Linearization} \label{sec:prob_fbl}
We aim to find $u_k$ such that the flat input $v_k$ in \eqref{eq:dt_drunowsky} seen by the system closely matches the desired flat input $v^*_k$, optimized in FMPC. 
We select the input $u_k$ that minimizes the expected square distance between the desired flat input $v^{*}_k$ and the output of the GP model for \eqref{eq:v_from_u} as
\begin{equation} \label{eq:sf_expect}
\min_{u_k} \Expectation{\| \psi(\Z^*_k, u_k) - \vdes_k \|^2}.
\end{equation} 
When we query the GP model of \eqref{eq:nonlin_dyn} at $\Feature = (\Z^*_k,u_k)$ the posterior prediction of \eqref{eq:v_from_u} is normally distributed $\psi(\Z^*_k, u_k) \vert \Data = \NormalDist{\mean(\Z^*_k,u_k)}{\cov^2(\Z^*_k,u_k)}$. 
Consequently, \eqref{eq:sf_expect} can be written as $\min_{u_k} (\mean(\Z^*_k,u_k) - \vdes_k)^2 + \cov^2(\Z^*_k,u_k)$. 
Exploiting the affine form of the GP kernel selection allows for the mean and covariance to be substituted by \eqref{eq:aff_mean} and \eqref{eq:aff_cov}, respectively, further simplifying \eqref{eq:sf_expect} to
\begin{align} \label{eq:sf_min_full}
\min_{u_k} (\gamma^{*2}_2 + \gamma^*_5) u_k^2 + (2 \gamma^*_1 \gamma^*_2 - 2\gamma^*_2 \vdes + \gamma^*_4)u_k,
\end{align}
where $\gamma^*_i \coloneqq \gamma_i(\Z^*_k)$.
\begin{remark}
Following from \cite{greeffLearningStabilityFilter2021}, the optimization problem in \eqref{eq:sf_min_full} is convex since $\gamma_5^* \geq 0$ as it is the predicted covariance of $\beta(\Z)$ in \eqref{eq:v_from_u_a_b}, and the minimization is quadratic in the optimization variable $u_k$. 
\end{remark}

\subsubsection{Probabilistic Stability Constraints} \label{sec:stab}
We formulate a probabilistic stability constraint using the CLF from \eqref{eq:lyap_decrease} to ensure that the input $u_k$ guarantees probabilistic stability for the closed-loop system, despite \eqref{eq:nonlin_dyn} being unknown.

To formulate this constraint, consider the Lyapunov function of the form $V(\e_k) = \new{\e_k^T} \Pv \e_k$ and a nominal flat input $\vnom_k = -\K\e_k + \vref_k$, using the gain computed in \eqref{eq:fmpc_fb}. Then, the error at $k+1$ can be expressed as $\e_{k+1} = \Ad \e_k - \Bd \K \e_k + \Bd (\PSI - \vnom_k)$.
After using the discrete-time algebraic Ricatti equation, the CLF decrease condition \eqref{eq:lyap_decrease} can be expressed as
\begin{gather}
\e_k^T[\Pv - \Q - \R \K^T  \K] \e_k \\
 - 2 \e_k^T (\Ad - \Bd \K)^T \Pv\Bd (\PSI - \vnom_k) \label{eq:sf_mid} \\
 + (\PSI - \vnom_k)^2 \Bd^T \Pv \Bd \leq  \e_k^T \Pv \e_k\new{ - \epsilon,}
\end{gather}
\new{where $\epsilon > 0$ is a small constant to allow for the inequality to be non-strict.}
Thus, the left-hand side of \eqref{eq:sf_mid} is quadratic in $\psi(\Z^*_k,u_k)$ which we have learned as a GP.
We conservatively bound the last term $(\PSI - \vnom_k)^2 \Bd^T \Pv \Bd$ as the eigenvalues of $\Bd^T \Pv \Bd$ are proportional to $\timestep^2$, making this term small relative to the other terms.
Using the fact that the posterior mean prediction \eqref{eq:aff_mean} is affine in $u_k$ and that $u_k$ is bounded $\umin \leq u_k \leq \umax$,
\eqref{eq:sf_mid} can be written as
\begin{equation} \label{eq:sf_pre_prob}
-w_1 (\PSI - \vnom_k) \leq w_3 - w_2,
\end{equation}
where
\begin{gather}
w_1 \coloneqq 2\e_k^T (\Ad - \Bd \K)^T \Pv\Bd, \\
w_2 \coloneqq \Bd^T \Pv \Bd \max_{s = \{\umin, \umax\}} \|  \mean(\Z^*_k,s)- \vnom_k \|^2, \\
w_3 \coloneqq \e_k^T [\Q + \R \K^T  \K] \e_k \new{- \epsilon}.
\end{gather}

\begin{assumption} \label{ass:RKHS}
The nonlinear single-input control affine system \eqref{eq:nonlin_dyn} permits a bounded reproducible kernel Hilbert space (RKHS) norm $\| \psi(\Z_k,u_k) \|_\text{kern}$ with respect to the kernel \eqref{eq:aff_kern} used in the GP, and the GP's observation noise $\eta$ is uniformly bounded by $\sigma_\eta$.
\end{assumption}
\begin{lemma} \label{lem:bound}
Let $\delta \in (0,1)$. 
Given Assumption \ref{ass:RKHS}, 
$\Prob\{ -w_1(\mean(\Z^*_k,u_k) - \vnom_k) \leq w_3 - w_2 - | w_1 | \beta^{1/2} \cov(\Z^*_k,u_k)  \} \geq 1 - \delta$, 
where $\beta = 2 \| \psi(\Z^*_k,u_k) \|_{\text{kern}} + 300 \gamma \ln^3((N+1)/\delta)$,
and $\gamma \in \Real$ is the maximum information gain.
\end{lemma}

\begin{proof}
Given Assumption \ref{ass:RKHS} and  Theorem 3 in \cite{srinivasInformationTheoreticRegretBounds2012}, the mean prediction of a GP is bounded with respect to the true function evaluation as per
$\Prob \{ \forall \Feature \in \mathcal{A}, | \psi(\Feature) - \mean(\Feature) | \leq \beta^{1/2} \cov (\Feature ) \} \geq 1 - \SafeProb$.
When considering this probabilistic bound in the context of the left-hand side of \eqref{eq:sf_pre_prob} and expanding the absolute value, the upper inequality bound becomes
$-w_1 (\psi(\Z^*_k,u_k) - \vnom_k) \leq - w_1 (\mean(\Z^*_k,u_k) - \vnom_k) - | w_1 | \beta^{1/2} \cov(\Z^*_k,u_k).$
Using this expression with Theorem 3 from \cite{srinivasInformationTheoreticRegretBounds2012} means that the inequality holds true with probability $1 - \SafeProb$.
Using the probabilistic bound in \eqref{eq:sf_pre_prob}, yields the probabilistic constraint given in the lemma.
\end{proof}



\subsubsection{Probabilistic State Constraints} \label{sec:state}
We also ensure that $\Z_{k+1} \in  \Zcon$, with high probability, \new{by taking into account the uncertainty in the prediction of $\Z_{k+1}$ due to the uncertainty in the learned mapping $\psi(\Z^*_k,u_k)$}.
A constraint is required here beyond the state constraint in \eqref{eq:mpc} as the input can be modified by the safety filter, and the uncertainty in $\psi(\Z^*_k,u_k)$ must be accounted for.
We thus first determine the uncertainty in $\Z_{k+1}$, then use it to \new{tighten $\Zcon$ to} ensure $\Z_{k+1} \in  \Zcon$.

Using the posterior mean prediction from the GP in \eqref{eq:dt_drunowsky}, the mean of the next state becomes $\mean^z_{k+1} = \Ad \Z^*_k + \Bd \mean(\Z^*_k,u_k)$.
For brevity, we only consider uncertainty in $\mu^z_{k+1}$ due to $\cov(\Z^*_k,u_k)$.
Thus, the uncertainty in $\mu^z_{k+1}$ is given as $\cov^2_{\Z_{k+1}} = \Bd \cov^2(\Z^*_k,u_k) \Bd^\top$, which we use to tighten the constraint set $\Zcon$ using probabilistic reachable sets (PRS). 

\begin{definition}[One-Step PRS \cite{hewingCautiousModelPredictive2020a}]
Given the residual error of a random sample away from its mean $\PRSerror_k = \mu^Z_k - \Z_k$, a set $\PRS$ is called a One-Step PRS of probability level $\SafeProb$ if \new{$\Prob(\PRSerror_{k+1} \in \PRS | \PRSerror_k = 0) \geq \SafeProb$}.
\end{definition}

\begin{remark}
If we define an error term $\PRSerror^z_{k+1} = \mu^z_{k+1} - \Z_{k+1}$, we can define the tightened constraints on $\mu^z_{k+1}$ as 
\begin{equation} \label{eq:tight}
\mu^z_{k+1} \in \Zcon \ominus \PRS(\cov_{\Z_{k+1}}),
\end{equation}
where $\ominus$ represents the Pontryagin set difference.
\end{remark}

\begin{lemma}[Probabilistic Half-Space Constraints] \label{lem:half_space}
Consider a half-space constraint given by $\Zcon^{hs}  \coloneqq \{ \Z_k |\new{\Hcon}^T \Z_k \leq b \}$ with $\Hcon \in \mathbb{R}^n$ and $b \in \mathbb{R}_+$ defining the constraint.
Then, given the uncertainty in the dynamics propagation, the tightened constraint becomes
\begin{equation} \label{eq:z_tight}
\mathcal{Z}^{hs}(\cov_{\Z_{k+1}}) \coloneqq \left\{\Z_{k+1} |\new{\Hcon}^T \Z_{k+1} \leq b - \quantile(\SafeProb)\sqrt{\new{\Hcon}^T \cov_{\Z_{k+1}}^2\new{\Hcon}}\right\},
\end{equation}
which guarantees that the constraint will be satisfied given the uncertainty in the dynamics, with probability level $\SafeProb$.
Here, $\quantile$ is the quantile function of a standard Gaussian random variable.
\end{lemma}
\begin{proof}
Given the uncertainty in the next state, $\cov_{\Z_{k+1}}$ and under the random variable $\Delta^z_{k+1}$, the marginal distribution becomes $\Hcon^T\Delta^z_{k+1} \sim \NormalDist{0}{\new{\Hcon}^T \cov_{\Z_{k+1}}^2\new{\Hcon}}$.
Using the quantile function of a standard Gaussian $\quantile(\SafeProb)$ with probability level $\SafeProb$ a PRS can be constructed
\begin{equation}
\label{eq:z_prs}
\PRS^z(\sigma_z) \coloneqq \left\{ \PRSerror |\new{\Hcon}^T \PRSerror_{k+1}^z \leq \quantile(\SafeProb)\sqrt{\new{\Hcon}^T \cov_{\Z_{k+1}}^2\new{\Hcon}} \right\},
\end{equation}
as follows from \cite{hewingCautiousModelPredictive2020a}.
Using \eqref{eq:z_prs} in \eqref{eq:tight}, the probabilistic half-space constraint is as shown in Lemma \ref{lem:half_space}.
\end{proof}

\subsubsection{Safety Filter as an SOCP} \label{sec:socp}
Here, we formulate the safety filter as an SOCP.
%
%
\begin{theorem}[Second-Order Cone Program] \label{thm:socp}
Given Assumptions \ref{as:diff_flat}, \ref{as:cost_bound}, \ref{as:vf_bound}, \ref{ass:pd_kern}, and \ref{ass:RKHS}, the optimization problem given in \eqref{eq:sf_min_full} subject to the probabilistic asymptotic stability constraint given by Lemma \ref{lem:bound} and probabilistic state constraint Lemma \ref{lem:half_space} can be written as an SOCP
\begin{align} 
\min_{\usoc} \quad & [2\gamma_1^* \gamma_2^* - 2\gamma_2^* \vdes_k + \gamma_4^*, 1] \usoc, \\
\text{s.t.}\quad  & \new{\| \Asoc_i \usoc + \bsoc_i \| \leq \csoc_i^T \usoc + \dsoc_i \quad i \in \{1,2,3\},} \label{eq:full_socp} \\
& \umin \leq u_k \leq \umax,
\end{align}
where $\usoc = [u_k,q]^T$, $\Asoc_i \in \Real^{2 \times 2}$, $\bsoc_i \in \Real^2$, $\csoc_i \in \Real^{2}$, $\dsoc_i \in \Real$.
\end{theorem}
\begin{proof}
First, a dummy variable $q \geq ({\gamma_2^*}^2 + \gamma_5^*) u_k^2$ is introduced into the optimization problem to reformulate the cost as linear in $\usoc$.
Moreover, $0 \geq (\gamma_2^{*2} + \gamma_5^*)u_k^2 - 4q$ which can be rewritten as $(1+q)^2 \geq 4(\gamma_2^{*2} + \gamma_5^*)u^2 + (1-q)^2$. 
Note that both sides of this inequality are positive, so this can be rewritten
as a standard SOC constraint where
$\Asoc_1 = \Diag{2 \sqrt{ \gamma_2^{*2} + \gamma_5^*}, -1}$, $\bsoc_1 = [0,1]^T$, $\csoc_1 = [0, 1]^T$, and $\dsoc_1 = 1$.
The optimization objective then becomes $[ 2\gamma^*_1 \gamma^*_2 - 2\gamma_2^* \vdes + \gamma_4^*, 1] \usoc$.

Next, consider the probabilistic stability constraint \eqref{eq:sf_pre_prob}.
Using the posterior prediction mean and covariance expressions \eqref{eq:aff_mean} and \eqref{eq:aff_cov}, it is noted that 
\begin{equation} \label{eq:cov_soc}
\sqrt{\gamma^*_3 + \gamma^*_{4}u_k + \gamma^*_{5}u_k^2} = \left\| \begin{bmatrix}
\sqrt{\gamma^*_5}u_k + \frac{\gamma^*_{4}}{2 \sqrt{\gamma^*_{5}}} \\
\sqrt{\gamma^*_{3} - \frac{\gamma^{*2}_4}{4\gamma^*_{5}}} \end{bmatrix} \right\|_2.
\end{equation}
Using this, the stability constraint can be rewritten as an SOC constraint with
 $\Asoc_2 = \Diag{w_1 \sqrt{\gamma_5^*}, 0}$, the vectors 
$\bsoc_2 = \left[ w_1 \frac{\gamma_4^{*}}{2\sqrt{\gamma_5^*}}, w_1 \sqrt{\gamma_3^* - \frac{\gamma_4^{*2}}{4\gamma_5^*}}\right]^T $ and $\csoc_2 = [ \frac{w_1 \gamma_5^*}{\beta^{1/2}}, 0]^T$, and the scalar 
$\dsoc_2 = (w_1\gamma_2^* + w_3 - w_2)/\beta^{1/2}$.

Finally, the probabilistic state constraint from Lemma~\ref{lem:half_space} is transformed into an SOC constraint.
First, see that the term $\sqrt{\new{\Hcon}^T \cov_{\Z_{k+1}}^2\new{\Hcon}} = \sqrt{\new{\Hcon}^T \Bd \Bd^T\new{\Hcon} } \cov(\Z^*_k,u_k) $.
This has the same form as the stability constraint, allowing \eqref{eq:cov_soc} to be used to transform the state constraint into an SOC constraint with $\Asoc_3 = \Diag{w_s\sqrt{\gamma_5^*}, 0}$, $\bsoc_3 = \left[ w_s \frac{\gamma_4^*}{2\sqrt{\gamma_5^*}}, w_s \sqrt{\gamma_3^* - \frac{\gamma_4^{*2}}{4\gamma_5^*}} \right]^T$, $\csoc_3 = [-\Hcon^T \Bd \gamma_2^*, 0]^T$ and $\dsoc_3 =  -\Hcon^T\Ad\Z^*_k -\new{\Hcon}^T\Bd\gamma_1^* + b$, with $w_s = \quantile(\SafeProb)\sqrt{\Hcon^T \Bd \Bd^\new{T}\new{\Hcon} }$
\end{proof}
The full controller algorithm is presented in Algorithm \ref{alg:full}.
\new{\begin{remark}
The SOC form is maintained when uncertainty in the measured flat state $\Zhat_{k}$ is considered, but is not shown here for brevity.
\end{remark}}
\setlength{\textfloatsep}{0pt}
\begin{algorithm}[t]
\caption{Proposed control algorithm.}\label{alg:full}
 $k \rightarrow 0$. \\
 \new{\uIf{Training offline}{
 Train \new{the nonlinear map} $v = \psi(\Z,u)$.\\}}
\While{$k\timestep \leq  T$}{ 
    Measure the current flat state $\Zhat_k$.\\
    Find the optimal flat state and input $\Z^*_k, v^*_k$ by solving the OCP \eqref{eq:mpc}. \\
    Find $u^*_k$ that minimizes \eqref{eq:full_socp} using $\Z^*_k$ and $v^*_k$. \\
    Apply $u^*_k$ to the real system.\\
    Set $k \leftarrow k + 1$. \\
     \new{\uIf{Training online}{
       Update the nonlinear map with measured data. \\}}
}
\end{algorithm}

\vspace{-1.0em}
\section{Simulation}
\begin{figure*}
\vspace{1pt} 
    \centering
     \begin{subfigure}[b]{0.3\textwidth}
         \centering
         \includegraphics[width=\textwidth]{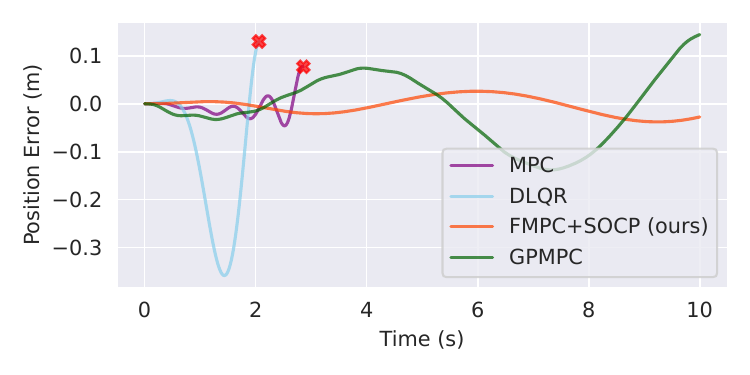}
         \caption{Tracking error for $y_\text{ref}(t) = 0.2t\sin(0.9t)$.}
         \label{fig:track}
     \end{subfigure}
     \hfill
     \begin{subfigure}[b]{0.3\textwidth}
         \centering
         \includegraphics[width=\textwidth]{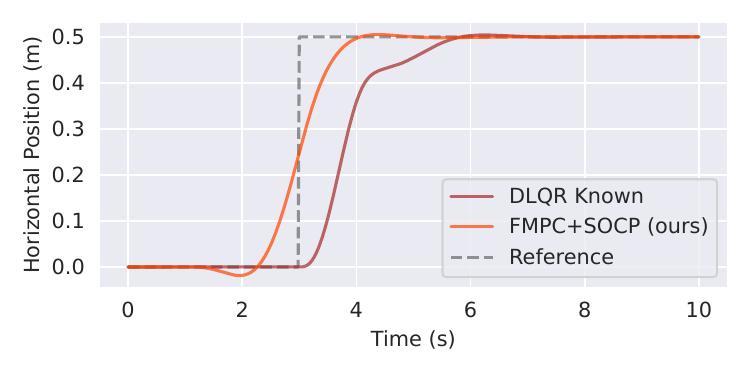}
         \caption{Tracking step with $-10^\circ \leq u_k \leq 10^\circ$.}
         \label{fig:step_in_con}
     \end{subfigure}
     \hfill
     \begin{subfigure}[b]{0.3\textwidth}
         \centering
         \includegraphics[width=\textwidth]{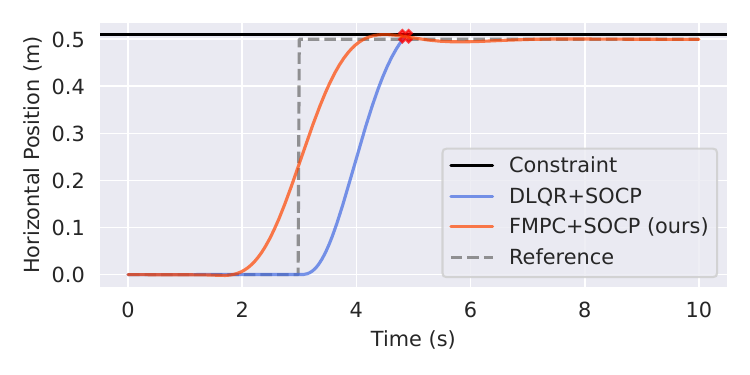}
         \caption{Step with constraint $x_k \leq 0.51$.}
         \label{fig:step_con}
     \end{subfigure} 
        \caption{In \ref{fig:track}, the tracking error of our Flat MPC with safety filter (FMPC+SOCP) is compared with an MPC and DLQR with inaccurate prior models, and a trained GPMPC \cite{hewingCautiousModelPredictive2020a}. In \ref{fig:step_in_con} the FMPC+SOCP is compared against a DLQR with perfect dynamics knowledge subject to input constraints, and in \ref{fig:step_con} FMPC+SOCP is compared with DLQR with the safety filter (DLQR+SOCP) subject to a constraint on the position.   
        Red $\times$ indicates a point of controller infeasibility or constraint violation. We see that FPMPC+SOCP performs similarly to GPMPC, while respecting input and state constraints.
        }
        \label{fig:three graphs}\
          \vspace{-1.5em}
\end{figure*}

\begin{figure}
         \centering
         \includegraphics[width=0.3\textwidth]{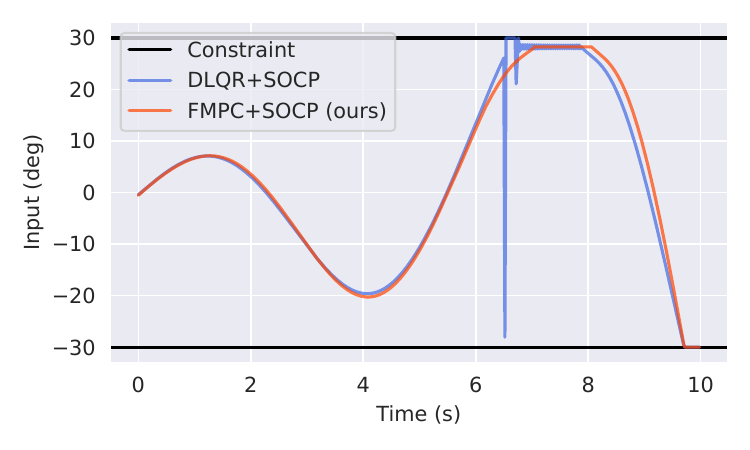}
         \vspace{-1.5em}
         \caption{ \new{Input comparison while tracking $y_\text{ref}(t) = 0.2t\sin(0.9t)$ subject to velocity constraints $\dot{x}_k \leq 1.0$ and input constraints $| u_k| \leq 30^{\circ}$.}}
         \label{fig:inputs}
\end{figure}

Our controller was assessed on three tasks using a horizontal 1-D quadrotor.
The quadrotor dynamics, as in \cite{Greeff2021} and \cite{greeffLearningStabilityFilter2021}, are given by $\xdd = \Gamma \sin(\theta)-\gamma \xd$ and $\dot{\theta} = \frac{1}{\tau}(u-\theta)$ where $x$ is the horizontal position, $\theta$ is the pitch angle of the quadrotor, and the commanded pitch angle $u$ is the system input.
Here, $ \Gamma=10$, $\gamma=0.3$ and $\tau=0.2$ are model parameters.
This model is flat in the output $y_k = x_k$, with the flat state $\Z_k = [x_k, \xd_k, \xdd_k]^T$.
All algorithms are run at 50 Hz, use the same gain matrices $\Q$ and $\R$, and horizons.

As shown in Figure \ref{fig:track}, FMPC with the safety filter (FMPC+SOCP) is compared against an MPC using an inaccurate model ($ \Gamma=20$, $\gamma=0$, and $\tau=0.05$, chosen to over-estimate thrust with no drag), and discrete linear quadratic regulator (DQLR) using the same inaccurate model, and a trained GPMPC where the uncertain dynamics are modelled as a GP inside a robust MPC formulation \cite{hewingCautiousModelPredictive2020a}.
Training data was gathered via \new{Latin hypercube sampling of states within the max and min values seen from the reference trajectories. Squared Exponential kernels were used for all GP kernels, trained offline. For this system, $v_k = \dddot{x}_k$, a higher-order derivative of the position, was measured from the simulation. On a real system, this would need to be estimated, which can be hard due to accumulated noise, presenting a limitation of our approach}.
The root mean squared errors were \qty{0.02}{m} and \qty{0.07}{m} for FMPC+SOCP and GPMPC, respectively. The average solve time per step was \qty[separate-uncertainty = true]{0.018(0.006)}{s} for FMPC+SOCP and \qty[separate-uncertainty = true]{0.29(0.04)}{s} for GPMPC \new{when run on a 16 GB RAM desktop using an AMD 3900xt CPU}.

The FMPC+SOCP evidently performs better than GPMPC, while solving the task more than 10 times faster.
This occurs in part because having the GP inside the nonlinear optimization can lead to worse local minima---something that nonlinear MPC is already prone to.
Additionally, the GP used in GPMPC requires an independent GP for each dimension of the state, meaning roughly five times as many training points were used in GPMPC than in our approach.

In Figure \ref{fig:step_in_con}, FMPC+SOCP is compared with DLQR where \eqref{eq:nonlin_dyn}, and thus \eqref{eq:v_from_u}, are known perfectly, but subject to input constraints $ | u_k | \leq 10^\circ$. 
The inclusion of DLQR is pertinent as it represents the infinite-horizon optimal solution.
In this case, DLQR's inputs are clipped if the desired input exceeds the limit.
We see that even in this case, the FMPC+SOCP outperforms the DLQR, as it is designed to account for input constraints.

Finally, FMPC+SOCP is compared against an inaccurate DLQR with the safety filter (DLQR+SOCP) on tracking a step trajectory with a state constraint $x_k \leq 0.51$ in Figure \ref{fig:step_con} \new{and tracking $y_\text{ref}(t) = 0.2t\sin(0.9t)$ subject to velocity constraints $\dot{x}_k \leq 1.0$ and input constraints $| u | \leq 30^{\circ}$ in Figure \ref{fig:inputs}}.
In Figure \ref{fig:step_con}, the predictive nature of the FMPC+SOCP clearly anticipates the step response and settles faster than the DLQR.
Additionally, DLQR+SOCP reaches an infeasible state as it overshoots the reference because it could not predict into the future nor account for the state constraint.
The FMPC+SOCP avoided violating the constraint boundary due to its predictive nature, remaining feasible.
\new{In Figure \ref{fig:inputs}, we see how the FMPC input anticipates the velocity limit, resulting in smooth inputs that never reach their limits. DLQR, however, gives very aggressive, and potentially damaging, changes in input to avoid violations.}

\vspace{-1.0em}
\section{Conclusion}
Current safe learning-based controllers' real-time applicability has been limited by computational performance.
The proposed FMPC+SOCP approach performs similarly to state-of-the-art learning-based controllers, but is 10 times more computationally efficient, while guaranteeing probabilistic asymptotic stability, probabilistic state constraint satisfaction, and input constraint satisfaction. 

\addtolength{\textheight}{-12cm}   




\vspace{-1.0em}
\bibliographystyle{IEEEtran} 
\bibliography{IEEEabrv,root}

\end{document}